\RequirePackage{fix-cm}
\documentclass{svjour3}                     
\smartqed  
\usepackage{graphicx,hyperref,amsfonts,amsmath}
\newtheorem{lem}{Lemma}
\newtheorem{defn}{Definition}
\newtheorem{thm}{Theorem}
\newtheorem{cor}{Corollary}
\newtheorem{rem}{Remark}

\newtheorem{exam}{Example}
\begin{document}

\title{Perfect edge state transfer on cubelike graphs
}


\author{Xiwang Cao 
}


\institute{Xiwang Cao \at
              Nanjing University of Aeronautics and Astronautics, Nanjing, Jiangsu Province, P.R. China\\
              Tel.: +86-025 52113704\\
              Fax: +86-025 52113704\\
              \email{xwcao@nuaa.edu.cn}           
}

\date{Received: date / Accepted: date}

\maketitle

\begin{abstract}
Perfect (quantum) state transfer has been proved to be an effective model for quantum information processing.  In this paper, we give a characterization of cubelike graphs having perfect edge state transfer. By using a lifting technique, we show that every bent function, and some semi-bent functions as well, can produce some graphs having PEST. Some concrete constructions of such graphs are provided. Notably, using our method, one can obtain some classes of infinite graphs possessing PEST.


\end{abstract}

\section{Introduction }

Quantum algorithm is the crucial part of quantum information processing and computation and is the research field of both mathematicians and  engineers around a few decades. In quantum mechanic, a qubit is the quantum analogue of a classical bit. Whereas a bit can take any value in
the set $\{0,1\}$, a qubit can be assigned to any $1$-dimensional subspace from a $2$-dimensional
complex vector space. A quantum state is represented by a vector in the complex vector space $\mathbb{C}^{\bigotimes 2^n}$ which is of the form $\sum_{v_1,\cdots, v_n\in \{0,1\}}a_{v_1\cdots v_n}|v_1\rangle \cdots |v_n\rangle$, $a_{v_1\cdots v_n}\in \mathbb{C}$. Given a graph $\Gamma=(V,E)$ with
$n$ vertices, where $V$ is the vertex set and $E$ is the edge set, we suppose that the vertices of the graph represent qubits, and that the edges
represent quantum wires between such qubits. The exact correspondence is built up as the following way:

To each vertex $v\in V$ we assign a qubit, that
is, a two-dimensional complex vector space $H_v\simeq \mathbb{C}^2$.
Thus the graph is associated to a space
isomorphic to $\mathbb{C}^{2^n}$. Denote the standard basis vectors of $\mathbb{C}^2$ by $|0\rangle$ and $|1\rangle$. For any subset $S$ of $V$, define
\begin{equation*}
  Q_S=\otimes_{u\in V}|i(u)\rangle, \mbox{ where $i(u)=\left\{\begin{array}{cc}
                                                           1, & \mbox{ if $u\in S$}, \\
                                                           0, & \mbox{ otherwise}.
                                                         \end{array}
  \right.$}
\end{equation*}
Thus $Q_s$ corresponds to a qubit state. For the error operators, we consider the following Pauli matrices:
\begin{equation*}
  \sigma^x=\left(
             \begin{array}{cc}
               0 & 1 \\
               1 & 0 \\
             \end{array}
           \right), \sigma^y=\left(
                               \begin{array}{cc}
                                 0 & -\imath \\
                                 \imath & 0 \\
                               \end{array}
                             \right), \sigma^z=\left(
                                                 \begin{array}{cc}
                                                   1 & 0\\
                                                   0 & -1 \\
                                                 \end{array}
                                               \right).
\end{equation*}
For a given ordering of the rows of the adjacency matrix $A$ of $\Gamma$, and $v\in V$, we define:
\begin{equation*}
  \sigma^x_v=I_2\otimes \cdots \otimes I_2\otimes \underset{{\small v\mbox{th position}}}{\sigma^x}\otimes I_2\otimes \cdots \otimes I_2,
\end{equation*}
where the product contains $n$ multiplicands. We also consider analogous definitions for $\sigma^y_v$
and $\sigma^z_v$.
The energy of the system is expressed in
terms of a Hermitian matrix $H$, called the Hamiltonian.
For a time-independent Hamiltonian:
\begin{equation*}
 H:= H_{xy}=\frac{1}{2}\sum_{uv\in E}\left(\sigma_u^x\sigma_v^x+\sigma_u^y\sigma_v^y\right).
\end{equation*}
The Schr${\rm\ddot{o}}$dinger equation of quantum mechanics will imply that the
evolution of the system is governed by the matrix $\exp(-\imath tH\hbar)$, where $\imath=\sqrt{-1}$, $t$ is a positive time
and $\hbar$ is the Planck constant divided by $2\pi$ \cite[p. 8 and p. 19]{coutinho}.

Upon certain choices of a time-independent Hamiltonian, more specifically the above $XY$-coupling model, the quantum system defined in certain graphs will evolve $\exp(\imath tA)$, where $A$ is the adjacency matrix of the corresponding graphs. In such a scenario, the dynamics of the quantum states in each vertex resembles in some aspects the dynamics of a random walk \cite{FG}.
Following this approach, Childs et al \cite{child} found a graph in which the continuous-time quantum walk (the concept will be defined below) spreads exponentially faster than any classical algorithm for a certain black-box problem. Childs also showed that the continuous-time quantum walk model is a universal computational model \cite{Ch2}.

Let $\Gamma=(V,E)$ be a simple graph (without loops and multiple edges) where $V$ is the vertex set and $E$ is the edge set. Let $A$ be the adjacency matrix of $\Gamma$, i.e.,
\begin{equation*}
  A=(a_{uv})_{u,v\in V}, \mbox{ where }a_{uv}=\left\{\begin{array}{cc}
                                    1, & \mbox{ if $(u,v)\in E$}, \\
                                    0, & \mbox{ otherwise}.
                                  \end{array}
  \right.
\end{equation*}
A continuous random walk on a graph is determined by a sequence of matrices of the form $M(t)$, indexed by the vertices of $\Gamma$ and parameterized by a real positive time $t$. The $(u,v)$-entry of $M(t)$ represents the probability of starting at vertex $u$ and reaching vertex $v$ at time $t$. Define a continuous random walk on $\Gamma$ by setting
\begin{equation*}
  M(t)=\exp(t(A-D)),
\end{equation*}
where $D$ is a diagonal matrix. Then each column of $M(t)$ corresponds to a probability density of a walk whose initial state is the vertex indexing the column.

For quantum computations, Fahri and Gutmann \cite{FG} defined an analogue continuous quantum walk, termed the {\it transfer matrix} of a graph $\Gamma$, as the following $n\times n$ matrix:
\begin{equation*}
  H(t)=H_{\Gamma}(t)=\exp(\imath tA)=\sum_{s=0}^{+\infty}\frac{(\imath tA)^s}{s!}=(H_{g,h}(t))_{g,h\in V},\ \ t\in \mathbb{R},
\end{equation*}
where $n=|V|$ is the number of vertices in $\Gamma$.
Suppose that the initial state of a walk is given by a density matrix $Q$ as physicists usually do. Then the state $Q(t)$ at time $t$ is given by
 $$Q(t)=H(t)QH(-t).$$
 We call a density matrix $Q$ a {\it pure state} if ${\rm rank}(Q) = 1$. We use $e_a$ to denote the standard basis vector in
$\mathbb{C}^n$ indexed by the vertex $a$. Then
$$Q_a = e_ae_a^t$$
is the pure state associated to the vertex $a$.

Physicists are interested in the question whether there exists a time $t$ such that for two distinct vertices
$a$ and $b$, it happens that $ Q_a(t) = Q_b$. When the above phenomenon occurs, we say that there is {\it perfect state transfer} (PST, in short) from $a$ to $b$ at the time $t$ in the graph.

Since $H(t)$ is a unitary matrix, if PST happens in the graph from $u$ to $v$, then  the entries in the $u$-th row and the entries in the $v$-th column are all zero except for the $(u,v)$-th entry. That is, the probability starting from $u$ to $v$ is absolutely $1$ which is an idea model for state transferring. The phenomenon of perfect state transfer in quantum communication networks was originally introduced by Bose in \cite{bose}. This work motivated much research interest and many wonderful applications of related works have been found in quantum information processing and cryptography (see \cite{ace,ahar,ahmadi,ber,chris1,chris2,Godsil1,Godsil2,Godsil3,ste,zhan} and the references therein.) In his three papers (\cite{Godsil1,Godsil2,Godsil3}), C. Godsil surveyed the art of PST and provided the close relationship
 between this topic and other researching fields such as algebraic combinatorics, coding theory etc. Ba${\rm \check{s}}$i${\rm \acute{c}}$  \cite{basic2} and Cheung \cite{che} presented a criterion on circulant graphs (the underlying group is cyclic) and cubelike graphs (the underlying group is $\mathbb{F}_2^m$) having PST. Remarkably, Coutinho et al \cite{cg} showed that one can decide whether a graph admits PST in polynomial time with respect to the size of the graph.
In a previous paper \cite{yingfengcao}, we present a characterization on connected simple Cayley graph $\Gamma={\rm Cay}(G,S)$ having PST, we give a unified interpretation of many previously known results. 

However, even though there are a lot of results on PST in literature, it is still quite rare in quantum walks. People are always interested in pursuing more graphs having PST.

Recently, the concept of perfect edge state transfer was introduced in \cite{chengodsil}. Instead of representing quantum state by density matrices associated with vertices, Chen and Godsil \cite{chengodsil} suggested to use the edges of a graph to index the density matrices. A graph is said to have {\it perfect edge state transfer} (PEST, in short) from an edge $(a,b)$ to edge $(c,d)$ if there
exists a complex scalar $\gamma$ with $|\gamma|=1$
 satisfying
$$H(t)(e_a-e_b) =\gamma
(e_c -e_d)$$
for some non-negative time $t$. In terms of the probability distribution,
there is perfect state transfer from $e_a-e_b$ to $e_c - e_d$ at time $t$ if and
only if
\begin{equation}\label{f-318}
 \left|\frac{1}{2}(e_a-e_b)^tH(t)(e_c-e_d)\right|^2=1.
\end{equation}

Edge state transfer shares a lot of properties with vertex state transfer and has potential applications in quantum information computation. Chen and Godsil\cite{chengodsil} provided a sufficient and necessary condition under which a graph has PEST. As applications of this characterization, Chen and Godsil proved the following results:
\begin{itemize}
  \item the path $P_n$ has PEST if and only if $n=3,4$;
  \item the cycle $C_n$ has PEST if and only $n=4$;
  \item if two graphs have PEST at the same time $t$, then so does their Cartesian product;
  \item if a graph has PEST, then its complement also has PEST;
  \item if a graph $\Gamma$ has PEST, then the join of graphs $\Gamma \Box \Delta$ also has PEST for some graphs $\Delta$.
\end{itemize}
In a recent paper \cite{luocao}, we proved that if a graph has PST, then it also has PEST.

However, up to date, there is no general characterization on which graphs have PEST. Even though there are some necessary and sufficient conditions for a graph to have PEST in \cite[Lemma 2.4,Thorem 3.9]{chengodsil} (see also Lemma \ref{chengodsil} in this paper), these conditions are not easy to be verified. Thus, a basic question for us is how to find simple and easily verified characterizations on graphs that have PEST. Especially for some particular graphs such as circulant graphs, cubelike graphs, Hamming graphs, etc. Moreover, more concrete constructions of graphs having PEST are always desirable.

In this paper, we present an explicit and easy to be verified condition for a cubelike graph to have PEST. See Lemma \ref{lem-1} and Theorem \ref{main-thm}. By taking a trace-orthogonal basis of $\mathbb{F}_{2^m}$ (the finite field of size $2^m$) over $\mathbb{F}_2$, one can check the above mentioned conditions just by evaluating the inner product of some vectors, including the calculating of the eigenvalues of the related graphs. Thus, Theorem \ref{main-thm} is a simplification of \cite[Lemma 2.4,Thorem 3.9]{chengodsil}. Moreover, by utilizing a so called ``lift technique", we construct some families of infinite cubelike graphs admitting PEST by involving (semi-)bent functions. The main idea of the technique is as follows: Since for a Cayley graph ${\rm Cay}(G,S)$, its eigenvalues are exactly the Walsh-Hadamard transformation of the characteristic function of the connection set $S$. In order to find graphs whose eigenvalues satisfy the condition of Theorem \ref{main-thm}, we embed a vector space into a larger vector space. Then we can suitably separate the eigenvalues by some planes. Moreover, in $\mathbb{F}_2^m$, every subset is one-to-one corresponding to a Boolean function, namely, the characteristic function. As a result, the eigenvalues of the graphs are represented by the Fourier spectra of the Boolean functions. By employing some specific Boolean functions, we can control the eigenvalues of graphs. Following this approach, we show that every bent function, and some semi-bent functions as well, can lead to a cubelike graph having PEST. See Theorem \ref{thm-4} and Theorem \ref{thm-semi}. Some concrete constructions are provided in Sect. \ref{examples}.

 We use $\mathbb{N}$, $\mathbb{Z}$, $\mathbb{Q}$, $\mathbb{R}$, and $\mathbb{C}$ to stand for the set of non-negative integers, the integers ring, rational numbers field, real numbers field and complex numbers field, respectively.

\section{Preliminaries}

In this section, we give some notation and definitions which are needed in our discussion.

\subsection{Characters group of an abelian group}
Let $G$ be a finite abelian group. It is well-known that $G$ can be decomposed as a direct product of cyclic groups:
\begin{eqnarray*}
  G=\mathbb{Z}_{n_1}\otimes \cdots\otimes \mathbb{Z}_{n_r}\ \ (n_s\geq 2),
\end{eqnarray*}
where $\mathbb{Z}_m=(\mathbb{Z}/m\mathbb{Z},+)$ is a cyclic group of order $m$.

For every $x=(x_1,\cdots,x_r)\in G$, $(x_s\in \mathbb{Z}_{n_s})$, the mapping
\begin{equation*}
   \chi_x:G\rightarrow \mathbb{C}, \chi_x(g)=\prod_{s=1}^r\omega_{n_s}^{x_sg_s} \ (\mbox{ for $g=(g_1,\cdots,g_r)\in G$})
\end{equation*}
is a character of $G$, where $\omega_{n_s}=\exp(2\pi i/n_s)$ is a primitive $n_s$-th root of unity in $\mathbb{C}$.
 For $x,y\in G$, we define $\chi_x\chi_y: G\rightarrow \mathbb{C}$ by
 \begin{equation*}
   \forall g\in G, (\chi_x\chi_y)(g)=\chi_x(g)\chi_y(g).
 \end{equation*}
 Then it can be shown that $\hat{G}=\{\chi_x| x \in G\}$ form a group which we call it the dual group or the character group of $G$. Moreover, the mapping $G\rightarrow \hat{G}, x\mapsto \chi_x$ is an isomorphism of groups. Furthermore, it is easy to see that
  $$\chi_x(g)=\chi_g(x) \mbox{ for all $x,g\in G$}.$$

\subsection{Trace-orthogonal basis}
Let $\mathbb{F}$ be field, $V_1,V_2$ be linear spaces over $\mathbb{F}$. A bilinear form over $\mathbb{F}$ is a two-variable function $B(x,y)$ on $V_1\times V_2$ satisfying:

(1) $B(ax+by,z)=aB(x,z)+bB(y,z), \forall x,y\in V_1, z\in V_2, a,b\in \mathbb{F}$;

(2) $B(x,ay+bz)=aB(x,z)+bB(y,z), \forall x\in V_1, y,z\in V_2, a,b\in \mathbb{F}$.

Let $V_1=V_2=\mathbb{F}_{2^m}$ which is viewed as a linear space over $\mathbb{F}_2$. Then it is easily seen that for $x,y\in \mathbb{F}_{2^m}$, $B(x,y):={\rm Tr}(xy)$ defines a bilinear form on $\mathbb{F}_{2^m}$, where ${\rm Tr}(\cdot)$ is the trace operator. A {\it trace-orthogonal basis} for $\mathbb{F}_{2^m}$ over $\mathbb{F}_2$ is a basis $\{\alpha_1,\cdots,\alpha_m\}$ satisfying:

\begin{equation*}
  {\rm Tr}(\alpha_i\alpha_j)=\left\{\begin{array}{cc}
                                      1 & \mbox{ if $i=j$}, \\
                                      0 & \mbox{ if $i\neq j$}.
                                    \end{array}
  \right.
\end{equation*}
It is known that for every finite field of even characteristic, there always exists a trace-orthogonal basis, see \cite{SL}. Let $x,y\in \mathbb{F}_{2^m}$ and suppose that $x=\sum_{i=1}^mx_i\alpha_i$, $y=\sum_{i=1}^my_i\alpha_i$, $x_i,y_i\in \mathbb{F}_2, i=1,\cdots,m$. Then
\begin{equation}\label{trace}
  {\rm Tr}(xy)=\sum_{i,j}x_iy_j{\rm Tr}(\alpha_i\alpha_j)=\sum_{i=1}^mx_iy_i.
\end{equation}

\section{A characterization of cubelike graphs having PEST}

Let $G$ be an abelian group with order $n$. Let $S$ be a subset of $G$ with $|S|=s\geq 1$, $0\not\in S=-S:=\{-z: z\in S\}$ and $G=\langle S\rangle$. Suppose that $\Gamma={\rm Cay}(G,S)$ is the Cayley graph with the connection set $S$. Take a matrix $P=\frac{1}{\sqrt{n}}(\chi_g(h))_{g,h\in G}$ and projection $E_x=p_xp_x^*$, where $p_x$ is the $x$-th column of $P$. Then the adjacency matrix $A$ of $\Gamma$ has the following spectral decomposition:
\begin{equation}\label{f-8}
  A=\sum_{g\in G}\lambda_gE_g,
\end{equation}
where
\begin{equation}\label{f-e6}
  E_x=p_xp_x^*=\frac{1}{n}(\chi_x(g-h))_{g,h\in G},
\end{equation}
and
\begin{equation}\label{f-9}
 \lambda_g=\sum_{s\in S}\chi_g(s), g\in G.
\end{equation}
Meanwhile, the transfer matrix $H(t)$ has the following decomposition:
\begin{equation}\label{f-10}
H(t)=\sum_{g\in G}\exp(\imath \lambda_g t)E_g.
\end{equation}
Thus we have, for every pair $u,v\in G$,
\begin{equation}\label{f-entry}
  H(t)_{u,v}=\sum_{g\in G}\exp(\imath \lambda_g t)(E_g)_{u,v}=\frac{1}{n}\sum_{g\in G}\exp(\imath \lambda_g t)\chi_g(u-v).
\end{equation}
Therefore,
\begin{eqnarray}\label{f-HE}
\nonumber&&\frac{1}{2}(e_c-e_d)^tH(t)(e_a-e_b)\\
\nonumber&=&\frac{1}{2}\left(H(t)_{c,a}-H(t)_{c,b}-H(t)_{d,a}+H(t)_{d,b}\right)\\
\nonumber&=&\frac{1}{2n}\sum_{x\in G}\exp(\imath t \lambda_x)\left(\chi_x(c-a)-\chi_x(c-b)-\chi_x(d-a)+\chi_x(d-b)\right)\\
\nonumber&=&\frac{2}{n}\sum_{x\in G}\exp(\imath t \lambda_x)\frac{\chi_x(c)-\chi_x(d)}{2}\frac{\overline{\chi_x(a)-\chi_x(b)}}{2}\\
&=&\frac{2}{n}\sum_{x\in G}\exp(\imath t \lambda_x)\chi_x(c-a)\frac{1-\chi_x(d-c)}{2}\frac{1-\overline{\chi_x(b-a)}}{2}.
\end{eqnarray}
In the sequel, we always assume that $(a,b), (c,d)$ are edges of the concerned graph. To avoid the trivial case, we assume that $b\neq c, a\neq d$.

Below, we let $G=(\mathbb{F}_{2^m}, +)$ be the additive group of the finite field $\mathbb{F}_{2^m}$. The character group of $G$ is
\begin{eqnarray*}
  \hat{G}=\widehat{(\mathbb{F}_q, +)}=\{\chi_z:z\in \mathbb{F}_q\},
\end{eqnarray*}
where for $g,z\in \mathbb{F}_q$, $\chi_z(g)=(-1)^{{\rm Tr}(zg)}$, and ${\rm Tr}: \mathbb{F}_q\rightarrow \mathbb{F}_2$ is the trace operator.
It is easy to see that $G\simeq \mathbb{F}_2^m$ which is an $m$-dimensional linear space over $\mathbb{F}_2$. If we view $\mathbb{F}_2^m$ as the additive group of the finite field $\mathbb{F}_{q}$ with $q=2^m$, then
\begin{eqnarray*}
  \hat{G}=\hat{\mathbb{F}_2^m}=\{\chi_z: z\in \mathbb{F}_2^m\},
\end{eqnarray*}
where for $g=(g_1,\cdots,g_m), z=(z_1,\cdots,z_m)\in \mathbb{F}_2^m$,
\begin{eqnarray*}
  \chi_z(g)=(-1)^{z\cdot g},\quad z\cdot g=\sum_{j=1}^mz_jg_j\in \mathbb{F}_2.
\end{eqnarray*}
The above two representations of additive characters of $\mathbb{F}_{2^m}$ can be unified by taking a trace-orthogonal
basis over $\mathbb{F}_{2^m}/\mathbb{F}_2$ (see (\ref{trace}).

In the finite field $\mathbb{F}_{2^m}$, it is known that the number of $x\in \mathbb{F}_{2^m}$ such that $\chi_x(\alpha)=1$ is $2^{m-1}$, where $\alpha\neq 0$ is any nonzero element in $\mathbb{F}_q$. Thus, the number of nonzero terms in the RHS of (\ref{f-HE}) is upper bounded by $2^{m-1}=n/2$. Moreover, the upper bound is achieved if and only if $d+c=b+a$. Furthermore, the absolute value of each term in the summation is less than or equal to 1. Therefore, we deduce that
\begin{equation*}
  \left|\frac{1}{2}(e_c-e_d)^tH(t)(e_a-e_b)\right|^2=1
\end{equation*}
if and only if the following conditions hold:

(1) $d+c=b+a$;

(2) for all $x\not\in {\ker}({\rm Tr}((b+a)X))$, $\exp(\imath t \lambda_x)\chi_x(c+a)$ is a constant.

Take an element $x_0\in \mathbb{F}_{2^m}$ satisfying ${\rm Tr}((b+a)x_0)=1$ and ${\rm Tr}((c+a)x_0)=0$. Then (2) is equivalent to $\exp(\imath t(\lambda_{x_0}-\lambda_x))=\chi_{x}(c+a)$ for all $x\in \mathbb{F}_{2^m}$ satisfying ${\rm Tr}((b+a)x)=1$.

Thus we have the following preliminary result:

\begin{lem}\label{lem-1}Let $\Gamma={\rm Cay}(\mathbb{F}_{2^m},S)$ be a cubelike graph over $\mathbb{F}_{2^m}$ with $|S|=s$. For $a,b,c,d\in \mathbb{F}_{2^m}$, $\Gamma$ has PEST between $(a,b)$ and $(c,d)$ if and only if

(1) $a+b+c+d=0$;

(2) Let $x_0\in \mathbb{F}_{2^m}$ such that ${\rm Tr}((b+a)x_0)=1$, ${\rm Tr}((c+a)x_0)=0$. Then $\exp(\imath t (\lambda_{x_0}-\lambda_x))={\chi_x(c+a)}$ for all $x\in \mathbb{F}_{2^m}$ satisfying $ {\rm Tr}((b+a)x)=1$.\end{lem}

Consequently, we have
\begin{cor}\label{cor-1}Let $\Gamma={\rm Cay}(\mathbb{F}_{2^m},S)$ be a cubelike graph over $\mathbb{F}_{2^m}$ with $|S|=s$. For $a,b,c,d,\alpha\in \mathbb{F}_{2^m}$, $\Gamma$ has PEST between $(a,b)$ and $(c,d)$ if and only if $\Gamma$ has PEST between $(a+\alpha,b+\alpha)$ and $(c+\alpha,d+\alpha)$.
\end{cor}

Define two subsets in $\mathbb{F}_{2^m}$ by
\begin{equation}\label{f-e2}
\begin{array}{l}
                                   \Omega_+=\{x\in \mathbb{F}_{2^m} : {\rm Tr}((c+a)x)=0,{\rm Tr}((b+a)x)=1\}, \\
                                   \Omega_-=\{x\in \mathbb{F}_{2^m}: {\rm Tr}((c+a)x)=1, {\rm Tr}((b+a)x)=1\}.
                                 \end{array}
\end{equation}
It is easily seen that
\begin{equation*}
  \Omega_+\cap \Omega_-=\emptyset, \ \ \Omega_+\cup \Omega_-=\{x\in \mathbb{F}_{2^m}: {\rm Tr}((b+a)x)=1\}.
\end{equation*}

If we take a trace-orthogonal basis of $\mathbb{F}_{2^m}$ over $\mathbb{F}_2$, then (\ref{f-e2}) has the following form:
\begin{equation}\label{f-e2'}
\begin{array}{l}
                                   \Omega_+=\{\mathbf{x}=(x_1\cdots x_m)\in \mathbb{F}_{2}^m : (\mathbf{c+a})\cdot \mathbf{x}=0,(\mathbf{b+a})\cdot \mathbf{x}=1\}, \\
                                   \Omega_-=\{\mathbf{x}=(x_1\cdots x_m)\in \mathbb{F}_{2}^m: (\mathbf{c+a}) \cdot \mathbf{x}=1, (\mathbf{b+a})\cdot \mathbf{x}=1\}.
                                 \end{array}
\end{equation}
where the ``$\cdot$" is the standard inner product of two vectors, and $\mathbf{a},\mathbf{b},\mathbf{c},\mathbf{d}$ are the vectors corresponding to $a,b,c,d$, respectively. Thus, in the very beginning, we can use the inner product to define the characters of $\mathbb{F}_{2}^m$ instead of using the trace mapping. This will make things easier for us. Only in the case of the isomorphism $\mathbb{F}_2^m\cong \mathbb{F}_{2^m}$ as linear spaces over $\mathbb{F}_2$ is concerned, we need to use the trace-orthogonal basis.

Recall that the $2$-adic exponential valuation of rational numbers which is a mapping defined by
\[
  v_2:\mathbb{Q}\rightarrow \mathbb{Z}\cup \{\infty\}, v_2(0)=\infty, v_2(2^\ell\frac{a}{b})=\ell, \mbox{ where $a,b,\ell\in \mathbb{Z}$ and $2\not |ab$}.
\]
We assume that $\infty+\infty=\infty+\ell=\infty$ and $\infty>\ell$ for any $\ell \in \mathbb{Z}$. Then $v_2$ has the following properties. For $\beta,\beta'\in \mathbb{Q}$,

(P1) $v_2(\beta\beta')=v_2(\beta)+v_2(\beta')$;

(P2) $v_2(\beta+\beta')\geq \min(v_2(\beta),v_2(\beta'))$ and the equality holds if $v_2(\beta)\neq v_2(\beta')$.

Based on the previous preparations, we present our main result as follows.

\begin{thm}\label{main-thm}Let $\Gamma={\rm Cay}(\mathbb{F}_{2^m};S)$ be a Cayley graph over $\mathbb{F}_{2^m}$ with $|S|=s$. For $a,b,c,d\in \mathbb{F}_{2^m}$, $\Gamma$ admits PEST between $(a,b)$ and $(c,d)$ if and only if the following two conditions hold:

(1) $a+b+c+d=0$;

(2) Let $x_0\in \Omega_+$. Then for all $x\in \Omega_-$, $v_2(\lambda_{x_0}-\lambda_x)$ are the same number, say $\rho$. Moreover, for every $y\in \Omega_+$, we have $v_2(\lambda_{x_0}-\lambda_y)\geq \rho+1$, where $\Omega_+,\Omega_-$ are defined by (\ref{f-e2}).

Furthermore, if the conditions (1), (2) are satisfied, then the time $t$ at which the graph has PEST is $t=\frac{(2u+1)\pi}{M}$ for some integers $u$, where $M=\gcd(\lambda_{x_0}-\lambda_x: x_0\neq x\in \mathbb{F}_{2^m}, {\rm Tr}((a+b)x)=1)$.  \end{thm}

\begin{proof}The main idea of the proof is an analogue of \cite[Theorem 2.4]{yingfengcao}. If $\Gamma$ has PEST at time $t$, then for $x,x'\in \Omega_-$, we have
\begin{equation}\label{f-e7}
 \exp(\imath t (\lambda_{x_0}-\lambda_x))={\chi_x(c+a)}, \ \ \exp(\imath t (\lambda_{x_0}-\lambda_{x'}))={\chi_{x'}(c+a)}.
\end{equation}
Write $t=2\pi T$. Then (\ref{f-e7}) becomes
\begin{equation}\label{f-e8}
  T(\lambda_{x_0}-\lambda_x)-\frac{1}{2}\in \mathbb{Z}, T(\lambda_{x_0}-\lambda_{x'})-\frac{1}{2}\in \mathbb{Z}.
\end{equation}
Thus, $T\in \mathbb{Q}$ and $T\neq 0$. Therefore, $v_2(T(\lambda_{x_0}-\lambda_x))=v_2(T(\lambda_{x_0}-\lambda_{x'}))=-1$ and then $v_2(\lambda_{x_0}-\lambda_x)=v_2(\lambda_{x_0}-\lambda_{x'})=-1-v_2(T).$ That is, for all $x\in \Omega_-$, $v_2(\lambda_{x_0}-\lambda_x)$ is a constant. Say, $\rho$.

For every $y\in \Omega_+$, condition (2) of Lemma \ref{lem-1} means that
\begin{equation*}
  \exp(\imath 2\pi T(\lambda_{x_0}-\lambda_y))=\chi_y(c+a)=(-1)^{{\rm Tr}((c+a)y)}=1.
\end{equation*}
Thus, $T(\lambda_{x_0}-\lambda_y)\in \mathbb{Z}$. Then $v_2(T(\lambda_{x_0}-\lambda_y))\geq 0$, i.e, $v_2(\lambda_{x_0}-\lambda_y)\geq -v_2(T)=\rho+1$.

Conversely, if for all $x\in \Omega_-$, $v_2(\lambda_{x_0}-\lambda_x)=\rho$ and for every $y\in \Omega_+$, $v_2(\lambda_{x_0}-\lambda_y)\geq \rho+1$, then
\begin{equation*}
  \exp(\imath t (\lambda_{x_0}-\lambda_x))={\chi_x(c+a)}\Leftrightarrow T(\lambda_{x_0}-\lambda_x)-\frac{1}{2}\in \mathbb{Z},
\end{equation*}
and
\begin{equation*}
  \exp(\imath t (\lambda_{x_0}-\lambda_y))={\chi_y(c+a)}\Leftrightarrow T(\lambda_{x_0}-\lambda_y)\in \mathbb{Z}.
\end{equation*}
Using the same argument as \cite[Theorem 2.4]{yingfengcao}, we get the desired result.\qed
\end{proof}

For any two states $e_a-e_b$ and $e_c-e_d$, they are termed {\it strongly cospectral} in $\Gamma$
if and only if $E_x(e_a -e_b) = \pm E_x(e_c-e_d)$ holds for  all $x\in \mathbb{F}_{2^m}$.

Now we let $\bigwedge^+_{ab,cd} $
denote the set of eigenvalues such that
$$E_x(e_a- e_b) = E_x(e_c - e_d)$$
and let $\bigwedge^-_{ab,cd} $ denote the set of eigenvalues such that
$$E_x(e_a- e_b) =- E_x(e_c- e_d).$$
It is easy to see that
$\bigwedge_{a,b}=\bigwedge_{c,d}=\bigwedge^+_{ab,cd}\cup \bigwedge^-_{ab,cd}, \bigwedge^+_{ab,cd}\cap \bigwedge^-_{ab,cd}=\emptyset,$
where $\bigwedge_{a,b}$ (resp. $\bigwedge_{c,d}$) is the set of eigenvalues $\lambda_x$ such that $E_x(e_a-e_b)\neq 0$ (resp. $E_x(e_c-e_d)\neq 0$).

Using strongly cospectrality, Chen and Godsil \cite{chengodsil} derived a characterization of
perfect edge state transfer as follows.

\begin{lem}\cite[Lemma 2.4]{chengodsil}\label{chengodsil} Let $\Gamma=(V,E)$ be a graph and $(a, b), (c, d) \in  E$. Perfect edge
state transfer between $(a,b)$ and $(c,d)$ occurs at time $t$ if and only
if all of the following conditions hold.

(a) Edge states $e_a -e_b$ and $e_c -e_d$ are strongly cospectral. Let
$\lambda_0 \in \bigwedge^+_{ab,cd}$.

(b) For all $\lambda_x\in \bigwedge^+_{ab,cd}$,
there is an integer $k$ such that $ t(\lambda_0-\lambda_x) =2k\pi$.

(c) For all $\lambda_x\in \bigwedge^-_{ab,cd}$,
there is an integer $k$ such that $ t(\lambda_0-\lambda_x) =(2k+1)\pi$.
\end{lem}

One can use Lemma \ref{chengodsil} to give an alternative proof of Theorem \ref{main-thm}. Indeed, for the cubelike graph, we have $E_x=\frac{1}{n}(\chi_x(g-h))_{g,h\in \mathbb{F}_{2^m}}$. Thus
$E_x(e_a- e_b) = \pm E_x(e_c - e_d)$ if and only if ${\rm Tr}((b+a)x)={\rm Tr}((c+d)x)$ for all $x\in \mathbb{F}_{2^m}$. Therefore,  $e_a-e_b$ and $e_c-e_d$  are strongly cospectral in $\Gamma$
if and only if $a+b+c+d=0$. Moreover, let $t=2\pi T$. Then condition (b) in Lemma \ref{chengodsil} is equivalent to $T(\lambda_{x_0}-\lambda_x)\in \mathbb{Z}$ for $x\in \Omega_+$. Meanwhile, the condition (c) is equivalent to $T(\lambda_{x_0}-\lambda_x)\in \frac{1}{2}+\mathbb{Z}$ for $x\in \Omega_-$, where $\Omega_+, \Omega_-$ are defined in (\ref{f-e2}).

\section{A lower bound for the time at which a cubelike graph has PEST}

In view of Corollary \ref{cor-1}, we can assume that the initial state edge is $(0,b)$, where $b\neq 0$. Moreover, we have the following result.

\begin{lem}\label{lem-2}Let $\Gamma={\rm Cay}(\mathbb{F}_{2^m};S)$ be a cubelike graph over $\mathbb{F}_{2^m}$ with $|S|=s$. Let $b,c,d\in \mathbb{F}_{2^m}$ and $b\neq 0$. Denote a set $S'=b^{-1}S=\{b^{-1}z: z\in S\}$. Then $\Gamma={\rm Cay}(\mathbb{F}_{2^m};S)$ has PEST at time $t$ between $(0,b)$ and $(c,d)$ if and only if $\Gamma'={\rm Cay}(\mathbb{F}_{2^m},S')$ has PEST between $(0,1)$ and $(b^{-1}c,b^{-1}d)$ at time $t$.\end{lem}
\begin{proof}Firstly, it is obvious that $b+c+d=0\Leftrightarrow 1+b^{-1}c+b^{-1}d=0$.

For condition (2). The eigenvalues of $\Gamma'$ are
\begin{equation*}
  \lambda_x'=\sum_{z\in S'}\chi_x(z)=\sum_{z\in S}\chi_x(b^{-1}z)=\sum_{z\in S}\chi_{b^{-1}x}(z)=\lambda_{b^{-1}x}.
\end{equation*}
Suppose that $\Gamma$ has PEST between $(0,b)$ and $(c,d)$. Then by
substituting $x$ by $b^{-1}x$ in the condition (2), we have whenever $b^{-1}x\in  {\ker}({\rm Tr}(bX))$, i.e, $x\in {\ker}({\rm Tr}(X))$,
\begin{equation*}
 \exp(\imath t(\lambda_{x_0}-\lambda_x'))= \exp(\imath t (\lambda_{x_0}-\lambda_{b^{-1}x}))=\chi_{b^{-1}x}(c)=\chi_x(b^{-1}c).
\end{equation*}
Thus by Lemma \ref{lem-1}, $\Gamma'$ has PEST between $(0,1)$ and $(b^{-1}c, b^{-1}d)$.

The converse direction can be proved similarly.\qed
\end{proof}

In the following, we consider the minimum time $t$ at which a cubelike graph has PEST. We will provide a lower bound on the time $t$. As we will see that in the Sect. \ref{examples}, the lower bound is almost tight in some cases.

By Lemma \ref{lem-1} and Lemma \ref{lem-2}, in the next context, without loss of generality, we only consider whether $\Gamma$ has PEST between $(0,1)$ and $(c,d)$.

Define two subset of $\mathbb{F}_{2^m}$ by $T_i=\{x\in \mathbb{F}_{2^m}, {\rm Tr}(x)=i\}, i=0,1$. In the following context, we fix $x_0$ as an element in $\Omega_+$.

Firstly, we have the following result.

\begin{thm}\label{bound}Let $S$ be a subset of $\mathbb{F}_{2^m}$ with $0\not\in S$ and $\langle S\rangle=\mathbb{F}_{2^m}$. Suppose that $\Gamma={\rm Cay}(\mathbb{F}_{2^m},S)$ has PEST between two edges $(0,1)$ and $(c,d)$ at time $t$. Then $t$ is of the form $\frac{(2u+1)\pi}{M}, u\in \mathbb{Z}$, where $M=\gcd(\lambda_{x_0}-\lambda_x: x\in T_1)$, $s=|S|$. Consequently, the minimum time $t$ is $\frac{\pi}{M}$. Moreover, if there is an element $z_0(\neq 1)$ in $S$ such that $1+z_0\not\in S$, then $M$ is a power of $2$, say, $M=2^\ell$, and $\ell\leq \left\lfloor \frac{\log_2(2s(s+3))}{2}\right\rfloor$.\end{thm}
\begin{proof}

Note that the first statement on the minimum time $t$ such that $\Gamma$ has PEST is $\frac{\pi}{M}$ by Theorem \ref{main-thm}. We proceed to prove  that if there is an element $z_0(\neq 1)$ in $S$ such that $1+z_0\not\in S$, then $M=2^\ell$ is a power of $2$. Furthermore, if $\Gamma$ has PEST, then $\ell \leq \left\lfloor\frac{\log_2(2s(s+3))}{2}\right\rfloor$.

For any cubelike graph ${\rm Cay}(\mathbb{F}_{2^m},S)$, we claim that $M=\gcd(\lambda_{x_0}-\lambda_x: x\in T_1)$ is a divisor of $2^m$ if there is an element $z_0(\neq 1)$ in $S$ such that $1+z_0\not\in S$. Before going to prove the claim, we first prove the following:
\begin{equation}\label{f-e9}
  \sum_{x\in T_1}\chi_x(z)=\left\{\begin{array}{cc}
                                    2^{m-1}, & \mbox{ if } z=0, \\
                                    -2^{m-1}, & \mbox{ if } z=1,\\
                                    0, & \mbox{ if $z\neq 0,1$}.
                                  \end{array}
  \right.
  \end{equation}
  \begin{equation}\label{f-e9'}
   \sum_{x\in T_0}\chi_x(z)=\left\{\begin{array}{cc}
                                    2^{m-1}, & \mbox{ if $z\in \mathbb{F}_2$}, \\
                                    0, & \mbox{ otherwise}.
                                  \end{array}
  \right.
  \end{equation}
Note that the map $L: \mathbb{F}_{2^m}\rightarrow T_0; x\mapsto x^2+x$ is two-to-one and ${\rm Tr}(x)={\rm Tr}(x^2)$ for all $x\in \mathbb{F}_{2^m}$. Thus
\begin{equation*}
  \sum_{x\in T_1}\chi_x(z)=\frac{1}{2}\sum_{x\in \mathbb{F}_{2^m}}(-1)^{{\rm Tr}((x^2+x+\delta)z)}=\left\{\begin{array}{cc}
                                    2^{m-1}\chi_z(\delta), & \mbox{ if } z\in \mathbb{F}_2, \\
                                    0, & \mbox{ otherwise}.
                                  \end{array}
  \right.
\end{equation*}
where $\delta\in \mathbb{F}_{2^m}$ with ${\rm Tr}(\delta)=1$. This completes the proof of (\ref{f-e9}). (\ref{f-e9'}) can be proved similarly.

Now, we prove the claim.

Suppose that on the contrary, there is an odd prime $p$ which is divisor of $M$. Let $\lambda_{x_0}-\lambda_x=Mt_x$ for all $x\in T_1$, where $t_x\in \mathbb{Z}$. Then we compute
\begin{eqnarray*}
  M\sum_{x\in T_1}t_x=\sum_{x\in T_1}(\lambda_{x_0}-\lambda_x) =\lambda_{x_0}2^{m-1}-\sum_{z\in S}\sum_{x\in T_1}\chi_x(z) =(\lambda_{x_0}+1_S(1))2^{m-1},
\end{eqnarray*}
where $1_S$ is the characteristic function of $S$, i.e, $1_S(x)=1$ if $x\in S$ and $0$ otherwise.
We note that the last equality is based on (\ref{f-e9}) and the fact that $0\not\in S$.
Thus, $\lambda_{x_0}\equiv -1_S(1) \pmod p$, and then $\lambda_x\equiv -1_S(1) \pmod p$ for all $x\in T_1$.
 Now, for the element $z_0$ in $S$, $1+z_0\not\in S$, we have, on the one hand,
\begin{equation*}
  \sum_{x\in T_1}\lambda_x\chi_x(z_0)=\sum_{y\in S}\sum_{x\in T_1}\chi_x(y+z_0)= 2^{m-1} \mbox{ (by (\ref{f-e9}))}.
  \end{equation*}
  On the other hand,
  \begin{equation*}
  \sum_{x\in T_1}\lambda_x\chi_x(z_0)\equiv \sum_{x\in T_1}(-1_S(1))\chi_x(z_0)\equiv (-1_S(1))\sum_{x\in T_1}\chi_x(z_0)\equiv 0 \pmod p.
  \end{equation*}
 Consequently, we have $p|2^{m-1}$ which is a contradiction. This proves the claim. That is, $M$ is a power of $2$, say, $M=2^\ell$.

Now, assume that $\Gamma$ has PEST. Then
for every $x\in \Omega_-$, we have $v_2(\lambda_{x_0}-\lambda_x)=\rho$. Moreover, for $y\in \Omega_+$, $v_2(\lambda_{x_0}-\lambda_y)\geq \rho+1$. Where $\Omega_+$ and $\Omega_-$ is defined by (\ref{f-e2}), or (\ref{f-e2'}). Thus $M=\gcd(\lambda_{x_0}-\lambda_x: x\in T_1)=2^\rho$. So that we have $\rho=\ell$.

For $x\in \Omega_-$, write $\lambda_{x_0}-\lambda_x=2^\rho \zeta(x)$, where $\zeta(x)\in \mathbb{Z}$ and $2\not|\zeta(x)$. Then
\begin{equation}\label{f-e10}
  \sum_{x\in T_1}\lambda_x^2=\sum_{y,z\in S}\sum_{x\in T_1}\chi_x(y+z)=2^{m-1}\left(s-|\{z: z\in S,1+z\in S\}|\right),
\end{equation}
Moreover,
\begin{equation*}\label{f-e11}
 \sum_{x\in T_1}(\lambda_{x_0}-\lambda_x)^2=2^{m-1}(\lambda_{x_0}^2+2\lambda_{x_0}\cdot1_S(1)+s-|\{z: z\in S,1+z\in S\}|).
\end{equation*}
As a consequence, we have
\begin{equation}\label{f-e12}
  \sum_{x\in \Omega_-}(\lambda_{x_0}-\lambda_x)^2\leq \sum_{x\in T_1}(\lambda_{x_0}-\lambda_x)^2\leq 2^{m-1}(s^2+3s).
\end{equation}
Meanwhile,
\begin{equation}\label{f-e13}
  \sum_{x\in \Omega_-}(\lambda_{x_0}-\lambda_x)^2=\sum_{x\in \Omega_-}2^{2\rho}\zeta(x)^2\geq 2^{2\rho}2^{m-2} \ \ (\mbox{ by $2\not| \zeta(x)$}).
\end{equation}
Combining (\ref{f-e12}) and (\ref{f-e13}) together, we get
\begin{equation*}
  2^{2\rho}\leq 2s(s+3).
\end{equation*}
That is
\begin{equation*}
  \ell=\rho\leq \left\lfloor\frac{\log_2(2s(s+3))}{2}\right\rfloor.
\end{equation*}
This completes the proof.\qed
\end{proof}

\section{Bent functions and PEST}
Let $f: \mathbb{F}_2^m\rightarrow \mathbb{F}_2$ be a Boolean function. If we endow the vector space $\mathbb{F}_2^m$ with the structure of the finite field $\mathbb{F}_{2^m}$, thanks to the choice of a basis of $\mathbb{F}_{2^m}$ over $\mathbb{F}_2$, then every non-zero Boolean function $f$ defined on $\mathbb{F}_{2^m}$ has a (unique) trace expansion of the form:
\begin{equation*}
  f(x)=\sum_{j\in \Gamma_m}{\rm Tr}^{o(j)}_1(a_jx^j)+\epsilon (1+x^{2^m-1}), \forall x\in \mathbb{F}_{2^m},
\end{equation*}
where $\Gamma_m$ is the set of integers obtained by choosing one element in each cyclotomic coset of 2
modulo $2^m-1$, $o(j)$ is the size of the cyclotomic coset of 2 modulo $2^m-1$ containing $j$, $a_j\in \mathbb{F}_{2^{o(j)}}$
and $\epsilon = wt(f)$ modulo 2 where $wt(f)$ is the Hamming weight of the image vector of $f$, that is, the
number of $x$ such that $f(x) = 1$. Denote ${\rm supp}(f)=\{x\in \mathbb{F}_{2^m}: f(x)=1\}$. It is obvious that the map $f\mapsto {\rm supp}(f)$ gives a one-to-one mapping from the set of Boolean functions to the power set of $\mathbb{F}_{2^m}$.

The Walsh-Hadamard transform of $f$ is defined by
\begin{equation*}
  \widehat{f}(a)=\sum_{x\in \mathbb{F}_{2^m}}(-1)^{f(x)+{\rm Tr}(ax)}, \forall a\in \mathbb{F}_{2^m}.
\end{equation*}
\begin{defn}Let $m=2k$ be a positive even integer. A Boolean function $f$ is called bent if $\widehat{f}(a)\in \{\pm 2^{k}\}$  for all $a\in \mathbb{F}_{2^m}$.\end{defn}
\begin{defn}Let $m$ be a positive integer. A Boolean function $f$ from $\mathbb{F}_{2^m}$ to $\mathbb{F}_2$ is said to be semi-bent if $\widehat{f}(a)\in \{0,\pm 2^{\lfloor \frac{m}{2}\rfloor}\} $. \end{defn}
It is well-known that bent function exists on $\mathbb{F}_2^{2k}$ for every $k$. Moreover, for every positive integer $k$, there are many infinite family of such functions. Semi-bent functions also exist in $\mathbb{F}_{2^m}$ for all integer $m\geq 1$. See for example, \cite{sihem}.

Let $f: \mathbb{F}_2^m\rightarrow \mathbb{F}_2$ be a Boolean function. $S={\rm supp}(f)\subseteq \mathbb{F}_{2^m}$. For the cubelike graph $\Gamma={\rm Cay}(\mathbb{F}_{2^m},S)$, its eigenvalues are
\begin{equation}\label{f-e14}
 \lambda_x=\sum_{z\in S}\chi_x(z)=\sum_{y\in \mathbb{F}_{2^m}}\frac{1-(-1)^{f(y)}}{2}\chi_x(y)=-\frac{1}{2}\widehat{f}(x), 0\neq x\in \mathbb{F}_{2^m}.
\end{equation}

For bent and semi-bent functions, the following results are known. See for example \cite{sihem}, page 72 and page 422.

\begin{lem}\label{lem-4} Let $m=2k$ be a positive even integer. Let $f$ be a Boolean function on $\mathbb{F}_{2^m}$, $S={\rm supp}(f)$.

(1) If $f$ be a bent function, then $|S|=2^m\pm 2^{k-1}$. Moreover, define a function $\widetilde{f}$, called the dual of $f$, by
\begin{equation*}
  \widehat{f}(x)=2^k(-1)^{\widetilde{f}(x)}.
\end{equation*}
Then $\widetilde{f}$ is also a bent function. As a consequence, the numbers of occurrences of $\widehat{f}(x)$ taking the values $\pm 2^k$ are $2^m\pm 2^{k-1}$ or $2^m\mp 2^{k-1}$.

 (2) If $f$ is a semi-bent function, then $|S|\in \{2^{m-1}, 2^{m-1}\pm 2^{k-1}\}$. Moreover, we have the following table.
 \begin{center}Table 1. Walsh spectrum
of semi-bent functions
$ f$ with $f(0)=0$
\begin{tabular}{|c|c|}
  \hline
  Value of $\widehat{f}(x),x\in \mathbb{F}_{2^m}$ & Frequency \\
    \hline
 $0$ & $2^{m-1}+2^{m-2}$ \\
   \hline
  $2^{k+1}$ & $2^{m-3}+2^{k-2}$ \\
    \hline
   $-2^{k+1}$ & $2^{m-3}-2^{k-2}$ \\
  \hline
\end{tabular}
\end{center} \end{lem}

It is known also that if $f$ is a bent function, then so is its complementary function, i.e, $g=1+f$. Obviously, if $|{\rm supp}(f)|=2^{m-1}+2^{k-1}$, then $|{\rm supp}(1+f)|=2^{m-1}-2^{k-1}$. 


If $m=2k+1$ is an odd number, then we have the following concrete constructions of cubelike graphs having PEST. In these constructions, we use a so-called ``lifting technique". Precisely, we first take a bent function or semi-bent function $f$ on $\mathbb{F}_{2}^{2k}$, then find its support $S$. We construct a subset in $\mathbb{F}_2^{2k+1}$ by
\begin{equation}\label{f-S}
  S'=\{(0,z): z\in {\rm supp}(f)\}\cup \{(1,z): z\in {\rm supp}(f)\}.
\end{equation}
We then show that one can find a flat to separate the plane $T_1$ into $\Omega_+$ and $\Omega_-$. The main idea of the ``lifting technique" is to divide the eigenvalues of the graphs into two parts, such that in one part of $T_1$, the corresponding eigenvalues are all zero. Then hopefully, we can get some desired graphs having PEST. The following two results, namely, Theorem \ref{thm-4} and Theorem \ref{thm-semi}, are obtained following this approach.

\begin{thm}\label{thm-4} Let $m=2k+1$, $k\geq 2$. Let $f$ be a bent function on $\mathbb{F}_2^{2k}$ satisfying $f(1,1,\cdots,1)=1$. Let $S={\rm supp}(f)\subseteq \mathbb{F}_2^{2k}$ and let $S'$ be defined in (\ref{f-S}).
Take
\begin{equation*}
  a=(00^{(2k)}), b=(1 1^{(2k)}),c=(1 0^{(2k)}),d=(01^{(2k)})\in \mathbb{F}_2^m.
\end{equation*}
Let $\Gamma={\rm Cay}(\mathbb{F}_2^m,S')$ be the cubelike graph associated with the connection set $S'$. Then $\Gamma$ has PEST between $(a,b)$ and $(c,d)$ at time $\frac{\pi}{2^k}$.
\end{thm}

\begin{proof}
By (\ref{f-e2'}) and the choice of $a,b,c,d$, it is obvious that
\begin{eqnarray*}
  \Omega_+ &=& \{ (x_1,\cdots,x_m)\in \mathbb{F}_2^m: x_1+\cdots+x_m=1,x_1=0\},\\
  \Omega_-&=& \{ (x_1,\cdots,x_m)\in \mathbb{F}_2^m: x_1+\cdots+x_m=1,x_1=1\}.
\end{eqnarray*}
The eigenvalues of $\Gamma$ are determined by the following:

For every $\mathbf{x}\in \mathbb{F}_2^{2k}$,
\begin{equation*}
  \lambda_{(0\mathbf{x})}=\sum_{\mathbf{z}\in S}(-1)^{0\cdot 0+\mathbf{z}\cdot \mathbf{x}}+\sum_{\mathbf{z}\in S}(-1)^{1\cdot 0+\mathbf{z}\cdot \mathbf{x}}=2\chi_{\mathbf{x}}(S).
\end{equation*}
where $\chi_{\mathbf{x}}(S)=\sum_{\mathbf{z}\in S}(-1)^{\mathbf{z}\cdot\mathbf{x}}=-2^{k-1}(-1)^{\widetilde{f}(\mathbf{x})}=\pm 2^{k-1}$ by (\ref{f-e14}). Moreover,
\begin{equation*}
  \lambda_{(1\mathbf{x})}=\sum_{\mathbf{z}\in S}(-1)^{0\cdot 1+\mathbf{z}\cdot \mathbf{x}}+\sum_{\mathbf{z}\in S}(-1)^{1\cdot 1+\mathbf{z}\cdot \mathbf{x}}=0.
\end{equation*}
Therefore, for every $x=(x_1\cdots x_m)\in \Omega_+$, $v_2(\lambda_{x_0}-\lambda_x)=v_2(\pm 2^k\pm 2^k)\geq k+1$. Meanwhile, for every $y=(y_1\cdots y_m)\in \Omega_-$, $v_2(\lambda_{x_0}-\lambda_y)=v_2(\pm 2^k)=k$. The desired result follows with Theorem \ref{main-thm}.\qed \end{proof}

We can also use some semi-bent functions to get cubelike graphs which have PEST. Before going to present our result, we need a lemma first.

\begin{lem}\label{lem-5} \cite{dillon} Let $m=2k+1$ be an odd integer, $k\geq 2$. Let $f(x)={\rm Tr}(x^{2^e+1})$, where $e$ is a positive integer satisfying $\gcd(e,m)=1$. Then
\begin{equation}\label{f-e15}
 \widehat{ f}(a)=\left\{\begin{array}{cc}
                          0, & \mbox{ if ${\rm Tr}(a)=0$},\\
                          \pm 2^{k+1}, & \mbox{ if ${\rm Tr}(a)=1$}.
                        \end{array}
 \right.
\end{equation}
\end{lem}
Now, we give our construction of cubelike graphs having PEST based on semi-bent functions.
\begin{thm}\label{thm-semi}Let $m=2k+1$, $k\geq 1$ be a positive integer. Let $f(x)={\rm Tr}(x^{2^e+1})$ which is a semi-bent function on $\mathbb{F}_{2^{2k+1}}$. Let $S={\rm supp}(f)$. Then $|S|=2^{2k}$. Define a subset of $\mathbb{F}_2\times \mathbb{F}_{2^m}$ by (\ref{f-S}).
Put
\begin{equation*}
  a=(00^{(m)}), b=(11^{(m)}),c=(0 1^{(m)}),d=(10^{(m)})\in \mathbb{F}_2^m.
\end{equation*}
Let $\Gamma={\rm Cay}(\mathbb{F}_2^{m+1},S')$ be the cubelike graph associated with the connection set $S'$. Then $\Gamma$ has PEST between $(a,b)$ and $(c,d)$ at time $\frac{\pi}{2^{k+1}}$.
\end{thm}
\begin{proof}The proof of this result is similar to that of Theorem \ref{thm-4}. We embed $\mathbb{F}_2\times \mathbb{F}_{2^m}$ into $\mathbb{F}_{2}^{m+1}$ as a vector space. Recall that the dual group of $\mathbb{F}_2\times \mathbb{F}_{2^m}$ is $\widehat{F_2}\times \widehat{\mathbb{F}_{2^m}}$, i.e, $\{\chi_{(u,x)}=\chi_u\chi_x: \chi_u\in \widehat{F_2}, \chi_x\in \widehat{\mathbb{F}_{2^m}}\}$. For every $(v,y)\in \mathbb{F}_2\times \mathbb{F}_{2^m}$, $\chi_{(u,x)}(v,y)=(-1)^{uv+{\rm Tr}(xy)}$. If we fixed a trace-orthogonal basis of $\mathbb{F}_{2^m}$ over $\mathbb{F}_2$, then ${\rm Tr}(xy)=\mathbf{x}\cdot\mathbf{y}$, where $\mathbf{x}$(resp. $\mathbf{y}$) is the vector corresponding to $x$ (resp. $y$). Since $b-a=(11\cdots 1)=d-c$, we have
\begin{equation*}
  \chi_{(u,x)}(b-a)= \chi_{(u,x)}(d-c)=(-1)^{u+\mathbf{x}\cdot (1\cdots 1)}=(-1)^{u+{\rm Tr}(x)}.
\end{equation*}
We note that under a trace-orthogonal basis $\{\alpha_1,\cdots,\alpha_m\}$ of $\mathbb{F}_{2^m}$ over $\mathbb{F}_2$, the vector $(1\cdots1)$ corresponds to $1$ in $\mathbb{F}_{2^m}$, i.e, $\sum_{i=1}^m\alpha_i=1$. This fact can be proved as the following:

Since for every fixed $i$, ${\rm Tr}(\alpha_i\sum_{j=1}^m\alpha_j)=\sum_{j=1}^m{\rm Tr}(\alpha_i\alpha_j)=1$, for every $x=\sum_{j=1}^mx_j\alpha_j\in \mathbb{F}_{2^m}$, we have
\begin{equation*}
  {\rm Tr}(x(\sum_{i=1}^m\alpha_i-1))=\sum_{j=1}^mx_j(\sum_{i=1}^m{\rm Tr}(\alpha_j\alpha_i)-1)=0.
\end{equation*}
Thus, $\sum_{j=1}^m\alpha_j=1$.

Therefore, in this case, we have
\begin{eqnarray*}
T_1&=&\{(u,x)\in \mathbb{F}_2\times \mathbb{F}_{2^m}: u+{\rm Tr}(x)=1\}, \mbox{ and },\\
  \Omega_- &=& \{ (0,x)\in \mathbb{F}_2\times \mathbb{F}_2^m: {\rm Tr}(x) =1\}\subseteq T_1,\\
  \Omega_+&=& \{ (1,x)\in \mathbb{F}_2\times \mathbb{F}_2^m: {\rm Tr}(x)=0\}\subseteq T_1.
\end{eqnarray*}
(See (\ref{f-HE}) also for the idea of how to chose $a,b,c,d$).

Since $\widehat{f}(0)=0$, $|S|=2^{m-1}$. By Lemma \ref{lem-5} and (\ref{f-e14}),
the eigenvalues of $\Gamma$ are
\begin{equation*}
  \lambda_{(00)}=|S'|=2|S|=2^m, \lambda_{(1\mathbf{x})}=0, \mbox{ $\forall \mathbf{x}$, and } \lambda_{(0\mathbf{x})}=\pm 2^{k+1}, \forall (0,\mathbf{x})\in \Omega_-.
\end{equation*}
Thus, for every $y\in \Omega_-$,  $v_2(\lambda_{x_0}-\lambda_y)=v_2(0\pm 2^{k+1})=k+1$, and for $x\in \Omega_+$, $v_2(\lambda_{x_0}-\lambda_x)=v_2(0-0)=\infty>k+1$. Thus we get the desired result by Theorem \ref{main-thm}.\qed
\end{proof}

\begin{rem}In the above Theorem \ref{thm-semi}, we just use a special semi-bent function to construct some cubelike graphs having PEST. Using the same idea, one can obtain many such graphs by using some other semi-bent functions.\end{rem}
\section{Some concrete constructions}\label{examples}
In this section, we present some examples to illustrate our results.

\begin{exam}\label{exam-1} Let $m=3$ and $S=\{(001),(110),(010), (101)\}\in \mathbb{F}_2^3$. Let $\Gamma={\rm Cay}(\mathbb{F}_2^3,S)$ be the cubelike graph associated with the connection set $S$. Then $\Gamma$ doesn't have PEST between $(000,001)$ and $(c,c+(001))$ for every $c\in \mathbb{F}_2^3$.
\end{exam}
\begin{proof} The eigenvalues of $\Gamma$ are
\begin{eqnarray*}
 && \lambda_{(000)}=4, \mbox{ and}\\
 &&\lambda_{(x_1x_2x_3)}=(-1)^{x_1+x_2}+(-1)^{x_2}+(-1)^{x_3}+(-1)^{x_1+x_3}, (x_1x_2x_3)\neq (000).
\end{eqnarray*}
Therefore,
\begin{equation*}
  \lambda_{(000)}=4, \lambda_{(011)}=-4, \lambda_{(001)}=\lambda_{(010)}=\lambda_{(100)}=\lambda_{(101)}=\lambda_{(110)}=\lambda_{(111)}=0.
\end{equation*}
Take $a=(000), b=(001), c=(111),d=(110)$. Then by (\ref{f-e2'}),
\begin{eqnarray*}
  \Omega_+ &=& \{(x_1x_2x_3)\in \mathbb{F}_2^3: x_1+x_2+x_3=0, x_3=1\}\\
  \Omega_-&=& \{(x_1x_2x_3)\in \mathbb{F}_2^3: x_1+x_2+x_3=1, x_3=1\}.
\end{eqnarray*}
Thus,
\begin{equation*}
   \Omega_+=\{(001),(111)\}, \ \ \Omega_-=\{(101),(011)\}.
\end{equation*}
It is easy to see that the condition (2) of Theorem \ref{main-thm} fails, and then $\Gamma$ doesn't have PEST between $(a,b)$ and $(c,d)$. In fact, using the same way, we can check that there is no element $z\in \mathbb{F}_2^3$ such that $\Gamma$ has PEST between $(000,001)$ and $(z,z+(001))$.\qed
\end{proof}
However, if we change the connection set $S$ in Example 1, then we may get a cubelike graph which has PEST, as the following example shows.
\begin{exam}\label{exam-2} Let $S=\{(001),(110),(010)\}\in \mathbb{F}_2^3$. Let $\Gamma={\rm Cay}(\mathbb{F}_2^3,S)$ be the corresponding cubelike graph. Take $a=(000)$, $b=(001)$, $c=(101)$, $d=(100)$. Then $\Gamma$ has PEST between $(a,b)$ and $(c,d)$.\end{exam}
\begin{proof} In this case, the eigenvalues of $\Gamma$ are
\begin{eqnarray*}
 && \lambda_{(000)}=3, \mbox{ and}\\
 &&\lambda_{(x_1x_2x_3)}=(-1)^{x_1+x_2}+(-1)^{x_2}+(-1)^{x_3}, (x_1x_2x_3)\neq (000).
\end{eqnarray*}
Hence, $\lambda_{(000)}=3$ and
\begin{equation*}
\lambda_{(011)}=-3, \lambda_{(001)}=\lambda_{(110)}=\lambda_{(100)}=1, \lambda_{(010)}=\lambda_{(101)}=\lambda_{(111)}=-1.
\end{equation*}
Take $a=(000), b=(001), c=(101),d=(010)$. Then by (\ref{f-e2'}) again,
\begin{eqnarray*}
  \Omega_+ &=& \{(x_1x_2x_3)\in \mathbb{F}_2^3: x_1+x_3=0, x_3=1\}\\
  \Omega_-&=& \{(x_1x_2x_3)\in \mathbb{F}_2^3: x_1+x_3=1, x_3=1\}.
\end{eqnarray*}
Thus,
\begin{equation*}
   \Omega_+=\{(101),(111)\}, \ \ \Omega_-=\{(001),(011)\}.
\end{equation*}
Therefore, we have $v_2(\lambda_{x_0}-\lambda_x)=1$ for $x\in \Omega_-$, and $v_2(\lambda_{x_0}-\lambda_x)\geq 2$ for $x\in \Omega_+$. By Theorem \ref{main-thm}, $\Gamma$ has PEST between $(a,b)$ and $(c,d)$ at time $t=\frac{\pi}{2}$.\qed \end{proof}

\begin{rem}Note that the characteristic function of the connection set in Example \ref{exam-2} is $$f(x_1x_2x_3)=x_1x_2x_3+x_1x_3+x_2+x_3$$ which is neither a bent function nor a semi-bent function on $\mathbb{F}_2^3$. The Walsh Hadamard spectra of it are $\{6^{(1)},2^{(4)}, -2^{(3)}\}$. Note also that in this example, we just use the inner product to define the characters of $\mathbb{F}_2^3$. We don't use a trace-orthogonal basis, thus in this example, $1=(001)\in S$. Therefore $(0,1)$ is an edge of $\Gamma$.\end{rem}

\begin{exam}\label{exam-3} Let $m=4$ and $f(x_1x_2x_3x_4)=x_1x_4+x_2x_3+x_1$ be a Boolean function from $\mathbb{F}_{2}^4$ to $\mathbb{F}_2$. Take $S={\rm supp}(f)$. Then there is no PEST in $\Gamma$. Define $S'=\{(0,z):z\in S\}\cup \{(1,z):z\in S\}\subseteq \mathbb{F}_2^5$ and let $\Gamma'={\rm Cay}(\mathbb{F}_2^5,S')$. Then $\Gamma'$ has PEST between two edges.\end{exam}

\begin{proof} 
The Walsh-Hadamard Transformation of $f$ is
\begin{eqnarray*}
  \widehat{f}(y_1y_2y_3y_4) &=& \sum_{x_1,x_2,x_3,x_4\in \mathbb{F}_2}(-1)^{x_1+x_1x_4+x_2x_3+x_1y_1+x_2y_2+x_3y_3+x_4y_4} \\
 &=& -\sum_{x_1,x_2,x_3}(-1)^{x_1+x_2x_3+x_1y_1+x_2y_2+x_3y_3}\sum_{x_4}(-1)^{x_4(x_1+y_4)} \\
   &=&- 2(-1)^{y_4+y_1y_4}\sum_{x_2}(-1)^{x_2y_2}\sum_{x_3}(-1)^{x_3(x_2+y_3)}\\
   &=&4(-1)^{y_4+y_1y_4+y_2y_3}.
\end{eqnarray*}
Thus $f$ is a bent function. It is obvious that
\begin{equation*}
  S=\{(1000),(0111),(0110),(1010),(1100),(1111)\}.
\end{equation*}
The eigenvalues of $\Gamma$ are:
\begin{eqnarray*}
  \lambda_\mathbf{x}&=&\sum_{\mathbf{z}\in S}(-1)^{\mathbf{z}\cdot\mathbf{x}}\\
  &=&(-1)^{x_1}+(-1)^{x_2+x_3+x_4}+(-1)^{x_2+x_3}+(-1)^{x_1+x_3}\\
  &&+(-1)^{x_1+x_2}+(-1)^{x_1+x_2+x_3+x_4}.\\
\end{eqnarray*}
Thus, we get the following table:
\begin{center}The eigenvalues of $\Gamma$
{\small
\begin{tabular}{|c|c|c|c|c|c|c|c|c|}
  \hline
  $\mathbf{x}$&(0000) & (0001) & (0010) & (0011) & (0100) & (0101) & (0110) & (0111) \\
  \hline
  $\lambda_{\mathbf{x}}$&6 & 2& -2 & 2 & -2 & 2& 2 & -2 \\
  \hline
  \hline
 $\mathbf{x}$& (1000) & (1001) & (1010) & (1011)& (1100) & (1101)& (1110) & (1111)\\
  \hline
$\lambda_{\mathbf{x}}$ & -2 & -2 & -2 & -2 & -2 & -2 & 2& 2 \\
  \hline
\end{tabular}}
\end{center}
 Moreover, if $a=(0000),b=(1111), c=(1000), d=(0111)$, then
\begin{eqnarray*}
  \Omega_+&=&\{(x_1x_2x_3x_4): x_1+x_2+x_3+x_4=1, x_1=0\},\\
 \Omega_-&=&\{(x_1x_2x_3x_4): x_1+x_2+x_3+x_4=1,x_1=1\}.
\end{eqnarray*}
That is
\begin{equation*}
  \Omega_+=\{(0100),(0010),(0001),(0111)\}, \Omega_-=\{(1000),(1011),(1110),(1101)\}
\end{equation*}
One can verify that the condition (2) of Theorem \ref{main-thm} fails in this case. Thus there is no PEST between $(a,b)$ and $(c,d)$. Similarly, by a direct checking, we can find that there is no PEST occurring in $\Gamma$ for different choices of $a,b,c,d$.

Now, we consider the graph $\Gamma'$. For every $\mathbf{x}\in \mathbb{F}_2^4$, we have
\begin{equation*}
  \lambda_{(0\mathbf{x})}=\sum_{(0,\mathbf{z})\in S'}(-1)^{\mathbf{z}\cdot \mathbf{x}}+\sum_{(1,\mathbf{z})\in S'}(-1)^{\mathbf{z}\cdot \mathbf{x}}=2\sum_{z\in S}(-1)^{\mathbf{z}\cdot \mathbf{x}}=2\lambda_{\mathbf{x}},
\end{equation*}
and
\begin{equation*}
  \lambda_{(1\mathbf{x})}=\sum_{(0,\mathbf{z})\in S'}(-1)^{\mathbf{z}\cdot \mathbf{x}}+\sum_{(1,\mathbf{z})\in S'}(-1)^{1+\mathbf{z}\cdot \mathbf{x}}=0.
\end{equation*}
If we take
\begin{equation*}
  a=(00000), b=(11111), c=(10000),d=(01111),
\end{equation*}
then
\begin{eqnarray*}
  \Omega_+'&=&\{(x_1x_2x_3x_4x_5)\in \mathbb{F}_2^5: x_1+x_2+\cdots+x_5=1, x_1=0\}, \\
  \Omega_-'&=&\{(x_1x_2x_3x_4x_5)\in \mathbb{F}_2^5: x_1+x_2+\cdots+x_5=1, x_1=1\}.
\end{eqnarray*}
Therefore,
\begin{equation*}
  \Omega_+'=\{(0\mathbf{x}): \mathbf{x}\in \Omega_+\cup \Omega_-\}, \Omega_-'=\{(1\mathbf{x}): \mathbf{x}\in \Omega_+\cup \Omega_-\}.
\end{equation*}
The desired result now follows from Theorem \ref{main-thm}.\qed
\end{proof}

In fact, one can show that the graph $\Gamma$ in the above Example \ref{exam-3} admits uniform mixing. We guess that if a cubelike graph has uniform mixing, then it cannot have PEST.

The following example is actually based on Theorem \ref{thm-semi}.

\begin{exam}Let the connection set $S$ be defined as
\begin{equation*}
  S=\{(0111),(0101),(0011),(0110),(1111),(1101),(1011),(1110)\} \subseteq \mathbb{F}_2^4.
\end{equation*}
Let $\Gamma={\rm Cay}(\mathbb{F}_2^4,S)$ be the cubelike graph with connection set $S$. Take
\begin{equation*}
  a=(0000), b=(1111), c=(1000), d=(0111).
\end{equation*}
Then $\Gamma$ has PEST between $(a,b)$ and $(c,d)$.
\end{exam}
\begin{proof}Under a trace-orthogonal basis of $\mathbb{F}_2^4$, we have
\begin{eqnarray*}
  T_1 &=& \{(x_1x_2x_3x_4)\in \mathbb{F}_2^4: x_1+x_2+x_3+x_4=1\}, \\
  \Omega_+ &=&  \{(x_1x_2x_3x_4)\in \mathbb{F}_2^4: x_1+x_2+x_3+x_4=1, x_1=1\},\\
 \Omega_- &=&  \{(x_1x_2x_3x_4)\in \mathbb{F}_2^4: x_1+x_2+x_3+x_4=1,x_1=0\}.
\end{eqnarray*}
Thus,
\begin{eqnarray*}
  T_1 &=& \{(1000), (0100),(0010),(0001), (0111),(1011),(1101),(1110)\}, \\
 \Omega_+&=&\{(1101),(1011), (1110),(1000)\}, \\
 \Omega_- &=& \{(0111),(0100),(0010), (0001)\}.
\end{eqnarray*}
Moreover, for every $(x_1x_2x_3x_4)\in \mathbb{F}_2^4$, the corresponding eigenvalue of $\Gamma$ is
\begin{eqnarray*}
  &&\lambda_{(x_1x_2x_3x_4)}\\
  &=&\sum_{\mathbf{s}\in S}(-1)^{\mathbf{s}\cdot \mathbf{x}}\\
  &=&(-1)^{x_2+x_3+x_4}+(-1)^{x_2+x_4}+(-1)^{x_3+x_4}+(-1)^{x_2+x_3}\\
  &&+(-1)^{x_1+x_2+x_3+x_4}+(-1)^{x_1+x_2+x_4}+(-1)^{x_1+x_3+x_4}+(-1)^{x_1+x_2+x_3}.
\end{eqnarray*}
So that we have
\begin{eqnarray*}
  &&\lambda_{(0111)}=4,\lambda_{(0100)}=\lambda_{(0010)}=\lambda_{(0001)}=-4, \\
  &&\lambda_{(1101)}=\lambda_{(1011)}=\lambda_{(1110)}=\lambda_{(1000)}=0.
\end{eqnarray*}
Thus for every $x\in \Omega_-$, $v_2(\lambda_{x_0}-\lambda_x)=v_2(\pm 4)=2$, and for every $y\in \Omega_+, v_2(\lambda_{x_0}-\lambda_y)=\infty$. Hence, $\Gamma$ has PEST between $(a,b)$ and $(c,d)$ at time $\frac{\pi}{4}$.

Below, for the convenience of the reader, we would like to give an explanation on how we obtain the connection set $S$. Firstly, we find a primitive polynomial of degree $3$ over $\mathbb{F}_2$. Here we choose $m(x)=x^3+x^2+1$. Suppose that $m(\alpha)=0$. Then $\mathbb{F}_8=\mathbb{F}(\alpha)$. It can be verified that $\{\alpha, \alpha^2,\alpha^4\}$ is a trace-orthogonal basis of $\mathbb{F}_8$ over $\mathbb{F}_2$. Take $f(x)={\rm Tr}(x^3)$. Then by Lemma \ref{lem-5}, $f(x)$ is a semi-bent function. One can check that ${\rm supp}(f)=\{1,\alpha^3, \alpha^5,\alpha^6\}$. Secondly, we compute the coordinates of the elements in ${\rm supp}(f)$ under the trace-orthogonal basis.
\begin{equation*}
 \left\{\begin{array}{ccc}
          1 &=& 1\alpha+1\alpha^2+1\alpha^4,\\
  \alpha^3&=& 1\alpha+0\alpha^2+1\alpha^4,\\
  \alpha^5 &=& 0\alpha+1\alpha^2+1\alpha^4,\\
  \alpha^6&=&1\alpha+1\alpha^2+0\alpha^4.
        \end{array}
  \right.
\end{equation*}
Thus, we have ${\rm supp}(f)=\{(111),(101),(011),(110)\}$. Using the ``lifting technique", we get the connection set $S$.\qed
 \end{proof}
\section*{Conclusion Remarks}

In this paper, we provide an explicit and tractable characterization on cubelike graphs having PEST (see Lemma \ref{lem-1} and Theorem \ref{main-thm}). By importing a so called ``lifting technique", we show that one can obtain cubelike graphs having PEST by using bent or semi-bent functions (see Theorem \ref{thm-4} and Theorem \ref{thm-semi}). In fact, we can also use some plateau functions to get graphs which have PEST, see Example \ref{exam-1}. Characterizing such graphs is one of our further research topics.


%
%



\end{document}